%% file: main.tex
\documentclass[letterpaper, 10 pt, conference]{ieeeconf}  
\IEEEoverridecommandlockouts                              
\usepackage{cite}
\usepackage{xargs}                      
\usepackage{clevethm}
\usepackage{mathematix}
\usepackage{stmaryrd}
\usepackage{microtype}
\usepackage{units}

\newlist{inlinelist}{enumerate*}{1}
\setlist*[inlinelist,1]{%
  label=(\roman*),
}

\newcommand{\iid}{i.i.d.}
\newcommand{\dimProbSpace}{k}
\newcommand{\psd}[1]{\mathbb{S}^{#1}_{\pplus}}
\newcommand{\pd}[1]{\mathbb{S}^{#1}_{\pplus\pplus}}
\newcommand{\amb}{\mathcal{A}} 
\newcommand{\ambTV}{\amb^{\ell_1}}
\newcommand{\simplex}{\Delta}
\newcommand{\rdkw}{r_{\mathrm{DKW}}}
\newcommand{\rmcd}{r_{\mathrm{M}}}
\newcommand{\adkw}{\alpha_{\mathrm{DKW}}}
\newcommand{\amcd}{\alpha_{\mathrm{M}}}
\newcommand{\argdot}{\hspace{0.18em}\cdot\hspace{0.18em}}
\newcommand{\Ac}{\bar{A}} 
\newcommand{\stabRegion}{\mathcal{S}}
\DeclareMathOperator*\spectral{\rho}%
\newcommand{\LDRS}[1]{\mbox{#1-\text{LDRS}}}
\newcommand{\LRS}{\text{LRS}}
\newcommand{\LSS}[1]{\mbox{#1-\text{LSS}}}

\renewenvironment{thebibliography}[1]{%
  \begin{oldthebibliography}{#1}%
    \setlength{\parskip}{0ex}%
    \setlength{\itemsep}{0ex}%
}%
{%
  \end{oldthebibliography}%
}

\crefname{apdx}{appendix}{appendices}
\Crefname{apdx}{Appendix}{Appendices}

\title{\LARGE \bf
Safe Learning-Based Control of Stochastic Jump Linear Systems: a Distributionally Robust Approach
}
\author{Mathijs~Schuurmans$^{\dagger}$, Pantelis~Sopasakis$^{\ddag}$ and Panagiotis~Patrinos$^{\dagger}$
\thanks{$^{\dagger}$M. Schuurmans and P. Patrinos are with the Department 
  of Electrical Engineering (\textsc{esat-stadius}), KU Leuven, 
  Kasteelpark Arenberg 10, 3001 Leuven, Belgium.
       Email: \texttt{\{mathijs.schuurmans, panos.patrinos\}@esat.kuleuven.be}}%
\thanks{$^{\ddag}$P. Sopasakis is with Queen's University Belfast, School of Electronics, Electrical Engineering and Computer Science,
               Centre for Intelligent Autonomous Manufacturing Systems, BT9 5AH, Northern Ireland, UK.
               Email: \texttt{p.sopasakis@qub.ac.uk}}%
\thanks{The work of the first and third authors was supported by: FWO projects: No. G086318N; No. G086518N; 
        Fonds de la Recherche Scientifique -- FNRS, the Fonds Wetenschappelijk Onderzoek 
        -- Vlaanderen under EOS Project No. 30468160 (SeLMA), 
        Research Council KU Leuven C1 project No. C14/18/068 and the 
        Ford--KU Leuven Research Alliance project No. KUL0023.}
}%
\begin{document}

\maketitle
\thispagestyle{empty}
\pagestyle{empty}

\begin{abstract}
\input{abstract}
\end{abstract}
\input{intro}
\input{Problem}
\input{bounds}
\input{Reformulation}
\input{Numerical}
\input{conclusion}
\bibliographystyle{ieeetr}
\bibliography{references}
\begin{appendix}
\input{appendix.tex}
\end{appendix}
\end{document}

%% file: abstract.tex
We consider the problem of designing control laws for stochastic jump linear systems where the disturbances are drawn randomly from a finite sample space according to an unknown distribution, which is estimated from a finite sample of \iid{} observations. We adopt a distributionally robust approach to compute a mean-square stabilizing feedback gain with a given probability. The larger the sample size, the less conservative the controller, yet our methodology gives stability guarantees with high probability, for any number of samples. Using tools from statistical learning theory, we estimate confidence regions for the unknown probability distributions (ambiguity sets) which have the shape of total variation balls centered around the empirical distribution. We use these confidence regions in the design of appropriate distributionally robust controllers and show that the associated stability conditions can be cast as a tractable linear matrix inequality (LMI) by using conjugate duality. The resulting design procedure scales gracefully with the size of the probability space and the system dimensions. Through a numerical example, we illustrate the superior sample complexity of the proposed methodology over the stochastic approach.

%% file: intro.tex
\section{Introduction} \label{sec:intro}
\subsection{Background and motivation}
The ever-decreasing costs of measuring, communicating and storing data
have led to a variety of opportunities to apply learning-based and data-driven
methodologies in control~\cite{rosolia2018learning,domahidi2011learning}.
These opportunities are of particular interest for systems with inherent 
stochastic uncertainty, as data-driven methodologies may be used to 
reduce conservativeness in controller design, while retaining 
safety guarantees.

A natural way of addressing this trade-off is by adopting a \textit{distributionally robust} approach~\cite{dupavcova1987minimax,shapiro2002minimax}, 
which is gaining popularity in many fields including machine learning~\cite{smirnova2019distributionally} and control~\cite{van2016distributionally,yang2018wasserstein}.
It provides a framework which inherently accounts for uncertainty on probability estimates by generalizing two opposing approaches of
\textit{stochastic} and \textit{robust} control~\cite{sopasakis2019risk}. 
Performance and safety guarantees of the former~\cite{patrinos2014stochastic} require full knowledge of the underlying probability 
distribution of involved random variables, which in practice is only available by approximation. The robust approach, on the other hand,
aims at providing guarantees in the worst possible realization of the uncertain variables.
This disregard for available statistical knowledge typically leads to overly conservative solutions or infeasibility.
By contrast, the distributionally robust framework imposes robustness only with respect to a given set of probability distributions,
often called \textit{ambiguity sets}. The challenge is to appropriately design this ambiguity set in order to make a suitable trade-off between safety and performance.

In the past few years, the stochastic optimization community has proposed several methods for building ambiguity sets from data,
and solving corresponding optimization problems~\cite{delage2010distributionally,bertsimas2018data}. One popular approach is to estimate the unknown distribution (e.g., by the \emph{empirical estimate})
and to construct the ambiguity set as the set of distributions within some statistical distance, such as the \emph{Wasserstein distance}~\cite{esfahani2018data,gao2016distributionally}
or \emph{$\phi$-divergences}~\cite{ben2013robust} from this estimate. 
In this paper, we follow this line of reasoning and restrict the considered class of ambiguity sets to be the $\ell_1$-norm ball around the 
empirical probability estimate. Many of the obtained results can however be extended to more general classes of convex ambiguity sets, given the appropriate modifications.  

\subsection{Main contributions}
Firstly, we propose a data-driven, distributionally robust design methodology for synthesizing static feedback control gains for stochastic jump linear systems,
which, for \emph{any} finite sample size grants mean-square stability to the closed-loop system at a given confidence level.

Secondly, we propose a reformulation of the resulting stability conditions and approximate it by a tractable linear matrix inequality (LMI), which avoids enumerating the extreme points of
the polytopic $\ell_1$-ambiguity set. We demonstrate the computational gains of this formulation and show that, in practice, the induced conservativeness is very limited.

\subsection{Notation}

Let $I_n$ be the $n\times n$ identity matrix and $1_n \dfn (1)_{i=1}^n$. 
Let the sets of symmetric positive definite and positive semi-definite $n \times n$ matrices
be denoted as $\pd{n}$ and $\psd{n}$, respectively. 
We denote by $\otimes$ the Kronecker product. 
For $x,y\in\Re$, we define $\1_{y}(x)$ to be equal to $1$ if $x=y$, and $0$ otherwise.
We denote the expectation operator by $\E[\argdot]$ and the probability simplex by $\simplex_k \dfn \{ p \in \Re^\dimProbSpace {}\mid{} p_i \geq 0,\, \sum_{i=1}^{\dimProbSpace} p_i = 1\}$.
We define $\N_{[a, b]} \dfn \{ i \in \N \mid a \leq i \leq b\}$. 
The spectral radius of a matrix A is denoted $\spectral(A)$. We denote the dimensions of a vector $x$ by $n_x$. $\ball_p(x, r)$ is the $\ell_p$-norm ball of radius $r$ around $x \in \Re^{n_x}$.
Finally, we denote the support function and indicator function of a set $C$ by 
$\sigma_{C}(\argdot)$ and $\delta_{C}(\argdot)$ respectively.  

%% file: Problem.tex

\section{Problem statement} \label{sec:problem-statement}
\subsection{Stabilizing control of stochastic jump linear systems}
This paper considers the control of  discrete-time stochastic jump linear dynamical systems
with random disturbances $w_t$:
\begin{equation} \label{eq:system-dynamics}
    x_{t+1} = A(w_t) x_t + B(w_t) u_t.
\end{equation}
The disturbances $w_t$ take values on the finite sample space $\mathcal{W}{}\dfn{}\N_{[1,k]}$ equipped with the discrete
$\sigma$-algebra $2^{\mathcal{W}}$. 
For all $i \in \N_{[1,\dimProbSpace]}$, we introduce the notation 
$A_i{}\dfn{}A(i)$ and $B_i{}\dfn{}B(i)$.
Furthermore, let 
\(\prob {}:{} 2^{\mathcal{W}} \rightarrow \Re\),
with \(\prob[w=i] {}={} \prob[\{i\}] {}={} p_i\)
be a probability measure, such that $(\mathcal{W}, 2^{\mathcal{W}}, \prob)$ defines a probability space.
Note that system~\eqref{eq:system-dynamics} is a specific type of Markov jump linear system (MJLS), 
where all the rows in transition probability matrix are identical. 
Consider now the analogously defined autonomous system 
\begin{equation} \label{eq:dynamics-autonomous}
    x_{t+1} = A(w_t) x_t,
\end{equation}
for which the following fundamental notion of stability is defined.

%
%

\begin{definition}[Mean-square stability{\cite[Def. 3.8]{costa2006discrete}}]
    We say that the autonomous system~\eqref{eq:dynamics-autonomous} is mean-square stable (MSS) if
    for each $x_0 \in \Re^{n_x}$: 
    \begin{inlinelist}
        \item  $\| \E[x_t] \| \rightarrow 0$; and
        \item $\| \E[x_t \trans{x_t}] \| \rightarrow 0 $,
    \end{inlinelist}
    as $t \rightarrow \infty$.
\end{definition}
This property can be verified by means of the following well-known conditions. 
\begin{thm}[Conditions for MSS] \label{thm:MS-stability-conditions}
    Defining the operator $T:\Re^{\dimProbSpace} \to \Re^{\dimProbSpace n_x^2 \times \dimProbSpace n_x^2}$ as  
    \begin{equation} \label{eq:operator-stability}
     T(p) \dfn (\trans{p} \otimes 1_\dimProbSpace \otimes I_{n_x^2}) \cdot \blkdiag(\{\trans{A}_i \otimes A_i\}_{i \in \N_{[1,\dimProbSpace]}}),
    \end{equation}
    the following statements are equivalent:
    \begin{enumerate}[label=(S\arabic*)]
        \item System~\eqref{eq:dynamics-autonomous} is MSS.
        \item $\spectral{(T(p)) < 1}$. \label{item:MS-condition-spectral}
        \item $\exists P \in \pd{n_x} {}:{} \sum_{i=1}^{\dimProbSpace} p_i \trans{A}_i P A_i - P \prec 0$.\label{item:MS-condition-Lyap}
    \end{enumerate}
\end{thm}
\begin{proof}
    Follows directly from~\cite[thm. 3.9 and Cor. 3.26]{costa2006discrete}.
\end{proof}

Ideally, our objective is to compute a linear state feedback gain $K$, such that the 
closed-loop system 
\begin{equation} \label{eq:closed-loop}
x_{t+1} = \Ac(w_t) x_t = (A(w_t) + B(w_t) K)x_{t} 
\end{equation}
is MSS. Unfortunately, however, application of ~\Cref{thm:MS-stability-conditions} requires the knowledge of $p$,
which is not available in practice. Instead, we assume to have access to a finite
sample $\{w_i\}_{i=1}^N$ of $N$ independent, identically distributed (\iid{}) disturbance values. 
We will show that it is possible to leverage non-asymptotic statistical information to 
design linear feedback laws which lead to a mean-square stable closed loop with \textit{high probability}.

\subsection{Mean-square stability in probability}

The proposed distributionally robust approach to certifying MSS in probability entails
the use of the available data to determine a nonempty, closed, convex set of probability distributions 
$\amb \subseteq \simplex_{\dimProbSpace}$ so that with high confidence,
$p \in \amb$ --- such a set is called an \emph{ambiguity set}~\cite{sopasakis2019risk}. The requirement that the closed-loop system~\eqref{eq:closed-loop} is 
MSS for all $\mu \in \amb$, leads to the \emph{distributionally robust} variant of the Lyapunov-type stability 
condition~\ref{item:MS-condition-Lyap}:
\begin{equation} \label{eq:DRMS-condition-Lyap}
    \exists P \in \pd{n_x} {}:{} \max_{\mu \in \amb} \sum_{i=1}^{\dimProbSpace} \mu_i \trans{\Ac_i} P \Ac_i - P \prec 0.
\end{equation}
Due to the dual representation of coherent risk measures~\cite[Thm. 6.4]{shapiro2009lectures}, the resulting property is equivalent
to risk-square stability with respect to the risk measure induced by $\amb$~\cite{sopasakis2019risk}.

Thus, given an ambiguity set $\amb$ which includes the true distribution $p$ at a given confidence level,
one can be equally confident that a controller for which the closed-loop system satisfies~\eqref{eq:DRMS-condition-Lyap},
is mean-square stabilizing.

The existence of such a controller depends on the system at hand. Therefore, it is useful to define the following
required property of the open-loop system \eqref{eq:system-dynamics}, which can be tested by feasibility of the problems 
described in~\Cref{sec:controller-design}.
\begin{definition}[Linear distributionally robust stabilizability] \label{def:LDRS}
We say that system~\eqref{eq:system-dynamics} is linearly distributionally robustly stabilizable
with respect to an ambiguity set $\amb$ (\LDRS{$\amb$}) if there 
exists a linear state feedback law $u(x) = Kx$ such that the corresponding closed-loop system~\eqref{eq:closed-loop}
is $p$-MSS for all $p \in \amb$.
\end{definition}
\begin{remark}
Based on~\Cref{def:LDRS}, we may additionally define \emph{linear robust stabilizability} (\LRS) of \eqref{eq:system-dynamics}
as \LDRS{$\simplex_{\dimProbSpace}$}, and \emph{linear stochastic stabilizability} with respect to the distribution $\hat{p}\in \simplex_{\dimProbSpace}$ (\LSS{$\hat{p}$}) 
as \LDRS{$\{\hat{p}\}$}. Since for any $\amb_1 \subseteq \amb_2$, $\LDRS{$\amb_2$} \Rightarrow \LDRS{$\amb_1$}$, \LRS{} and LSS can be viewed as the extreme cases of LDRS.
\end{remark}

\begin{remark}
    Provided that system \eqref{eq:system-dynamics} is \LRS, 
    the proposed approach can certify MSS with arbitrary confidence, regardless of the sample size. In contrast 
    to the robust approach, however, by collecting a (small) data sample, MSS can still be certified 
    when the system is only \LDRS{$\amb$} for some ambiguity set $\amb$. The required sample size is prescribed by the bounds described below.
    We illustrate this in \Cref{sec:numerical:vs-stochastic-robust}. 
\end{remark}

%% file: bounds.tex
\section{Learning-based ambiguity estimation} \label{sec:bounds}
Given $N$ independent samples $w_1, \ldots, w_N$ from the distribution of the disturbance, we
define the empirical measure $\hat{p} = (\hat{p}_i)_{i=1}^\dimProbSpace$, where
\begin{equation} \label{eq:empirical-estimate}
      \hat{p}_i 
{}\dfn{} 
      \tfrac{1}{N} {\textstyle \sum_{j=1}^{N} \1_{i}(w_j)},
\end{equation}
for all $i\in\N_{[1,\dimProbSpace]}$. We now derive upper bounds on the radius $r$ of the $\ell_1$-ambiguity set 
$\ambTV_r(\hat{p}) \dfn \{\mu \in \simplex_k \mid \| \mu - \hat{p} \|_1 \leq r\}$, such that for given $\alpha \in (0,1)$
\begin{equation} \label{eq:prob-inclusion}
\prob \big[ p \in \ambTV_r(\hat{p}) \big] \geq 1-\alpha. 
\end{equation}
Given such an ambiguity set, it is then possible to use the aforementioned stability condition \eqref{eq:DRMS-condition-Lyap}
to design controllers which are MS stabilizing with confidence $1-\alpha$. This is discussed further
in \Cref{sec:controller-design}.

\begin{figure}
\centering
\includegraphics{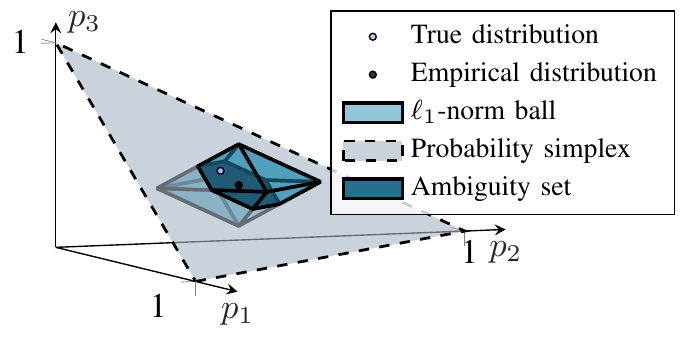}
\caption{Illustration of the $\ell_1$-ambiguity set $\ambTV(\hat{p})$ for $\dimProbSpace = 3$.}
\label{fig:ambiguity-example}
\end{figure}

\subsection{Dvoretzky-Kiefer-Wolfowitz bounds}\label{sec:DKW}
A statistical upper bound on $\lVert \hat{p}-p \rVert_1$ can be easily obtained by means of the 
\emph{Dvoretzky-Kiefer-Wolfowitz} (DKW) inequality~\cite{massart1990tight}, which probabilistically
bounds the error on the empirical estimate of the cumulative probability distribution.
We observe that this bound can be readily translated to the 
error on the probability distribution $p$.      
\begin{thm}[DKW ambiguity radius] \label{thm:radius-DKW}
Let $p, \hat{p} \in \simplex_\dimProbSpace$, $N$ be as defined in~\Cref{sec:bounds}. Then 
for any given confidence level $1-\alpha$, \cref{eq:prob-inclusion} holds with 
\begin{equation} \label{eq:radius-DKW}
  r=\rdkw(\alpha, \dimProbSpace, N) \dfn 2\dimProbSpace \sqrt{\tfrac{\ln{\nicefrac{2}{\alpha}}}{2N}}.
\end{equation}
\end{thm}
\begin{proof}
  The proof can be found in the~\cref{proof}\hspace{-0.4em}.
\end{proof}

\subsection{McDiarmid bounds}\label{sec:McDiarmid}
Alternatively, we may obtain a bound on the radius of the $\ell_1$-ambiguity set based on the following well-known 
measure concentration result. 
\begin{lem}[McDiarmid's inequality{\cite[Thm. 6.2]{boucheron2013concentration}}] \label{thm:McDiarmid:main}
If a function $f{}:{}\mathcal{W}^N \rightarrow \Re$ has the bounded differences property, i.e., there 
exist some constants $c_1, \ldots, c_N \geq 0$ such that, 
\begin{equation} \label{eq:bounded-differences}
\sup_{\substack{w_1,\ldots, w_N\\w_i' \in \mathcal{W}}}
  \lvert f(w_1,{\ldots}, w_N){-}f(w_1,{\ldots}, w_i',{\ldots} w_N)\rvert \leq c_i,
\end{equation}
and $\{w_i\}_{i=1}^N$ are independent random variables, then 
\begin{equation}
  \prob [f(w_1, \ldots, w_N) - \E[f(w_1, \ldots, w_N)] > \epsilon] \leq e^{\frac{-\epsilon^2}{2v}},
\end{equation}
where 
\(
  v \dfn \frac{1}{4} \sum_{i=1}^N c_i^2.
\)
\end{lem}

\begin{thm}  \label{thm:radius-McDiarmid}
    For a probability space of dimension $\dimProbSpace$, sample size $N$, and any given confidence level $1-\alpha$, \cref{eq:prob-inclusion} holds with 
  \begin{multline} \label{eq:radius-McDiarmid}
  r= \rmcd(\alpha, \dimProbSpace, N) \\\dfn \sqrt{-\frac{2 \ln(\alpha)}{N}} + \sqrt{\frac{2 (\dimProbSpace-1)}{\pi N}} + \frac{4 \dimProbSpace^{\nicefrac{1}{2}} (\dimProbSpace-1)^{\nicefrac{1}{4}}}{N^{\nicefrac{3}{4}}}. 
\end{multline} 
\end{thm}
\begin{proof}
First, we define a function $\psi: \mathcal{W}^N \rightarrow [0,2]$ and show 
that it satisfies the bounded differences condition~\eqref{eq:bounded-differences}:
\begin{align*}
    \psi(w_1, \ldots, w_N)  \dfn\lVert \hat{p} - p \rVert_1 
                             = \sum_{j=1}^{\dimProbSpace} \big \lvert \hat{p}_j (w_1,\ldots,w_N) - p_j \big \rvert.
\end{align*}
Due to the discrete support of \(w_i\), modifying
\(w_i = l \in \mathcal{W}\) to
\(w_i'= m \in \mathcal{W}\)
corresponds to decreasing $\hat{p}_l$, 
and increasing $\hat{p}_m$ by an amount \(\nicefrac{1}{N}\).
For ease of notation, we omit the function arguments and 
define $\psi' \dfn \psi(w_1,\ldots,w_i',\ldots, w_N)$ and $\hat{p}' \dfn \hat{p}(w_1,\ldots,w_i',\ldots, w_N)$, so that 
\begin{align*}
\psi' &= \| \hat{p}' - p \|_1 \\ 
        &= \lvert \hat{p}_l' - p_l \rvert + \lvert \hat{p}_m'  - p_m \rvert + \sum_{\substack{j=1,\\j\neq l,j\neq m}}^{\dimProbSpace} \lvert \hat{p}_j - p_j \rvert\\[-0.7em]
        &= \left \lvert \hat{p}_l - p_l - \tfrac{1}{N} \right \rvert + 
        \left \lvert \hat{p}_m  - p_m + \tfrac{1}{N} \right \rvert +
                \sum_{\substack{j=1,\\j\neq l,j\neq m}}^{\dimProbSpace} \lvert \hat{p}_j - p_j \rvert.
\end{align*}
Thus, \eqref{eq:bounded-differences} holds with \(c_{i} = \nicefrac{2}{N}\), and consequently \(v = \nicefrac{1}{N}\).
By \Cref{thm:McDiarmid:main}, then, 
\begin{equation} \label{eq:mcDiarmid:implicit}
 \prob \big[ \| \hat{p} - p \|_1 > \epsilon + \E[\| \hat{p} - p \|_1] \big] {}\leq{} e^{-\frac{N\epsilon^2}{2}}. 
\end{equation}
Moreover, from \cite[Lemma 7]{kamath2015learning}, we obtain a tight upper bound for the expected $\ell_1$-norm
of the estimation error:
\begin{equation}\label{eq:expectationTV}
\E\big[\lVert \hat{p} - p \rVert_1 \big] \leq \sqrt{\frac{2(\dimProbSpace-1)}{\pi N}} + \frac{4 \dimProbSpace^{\nicefrac{1}{2}}(\dimProbSpace-1)^{\nicefrac{1}{4}}}{N^{\nicefrac{3}{4}}}. 
\end{equation}
Using this result, \cref{eq:mcDiarmid:implicit} can easily be brought into the required form:
Let $\alpha = e^{-\frac{N \epsilon^2}{2}} \Rightarrow \epsilon = \sqrt{\nicefrac{2\ln(\nicefrac{1}{\alpha})}{N}}$, and substitute~\cref{eq:expectationTV}
into \eqref{eq:mcDiarmid:implicit} to obtain the bound~\eqref{eq:radius-McDiarmid}.
\end{proof}

The behavior of both bounds in terms of the sample size $N$ is similar; both decrease with $N$ as $\mathcal{O}(\nicefrac{1}{\sqrt{N}})$. 
In terms of $\dimProbSpace$, however, by virtue of~\eqref{eq:expectationTV}, $\rmcd \sim \mathcal{O}(\dimProbSpace^{\frac{3}{4}})$. This is an improvement to 
$\rdkw$, which increases linearly with $\dimProbSpace$. See~\Cref{sec:numerical:tightness} for a numerical comparison between these bounds and an empirical estimation of their tightness.

%% file: Reformulation.tex

\section{Design of distributionally robust controllers} \label{sec:controller-design}

We revisit the Lyapunov-type stability condition~\eqref{eq:DRMS-condition-Lyap}, and 
restate it in a slightly more general form that is more convenient when applied 
for constructing stabilizing terminal conditions for a receding horizon strategy. 

We denote the closed-loop dynamics corresponding to $w_t = i$ by $f_i(x, Kx) = A_i + B_i K x$, define 
$\ell(x,u) \dfn \trans{x} Qx{}+{}\trans{u} Ru$ with $Q{}\in{}\psd{n_x}$ and $R\in\pd{n_u}$, and denote 
the quadratic candidate Lyapunov function as $V(x) \dfn \trans{x} P x$.
Due to the homogeneity of \eqref{eq:DRMS-condition-Lyap}, we may replace the strict inequality by 
a non-strict inequality and introduce the negative definite quadratic form ${-\ell(x, Kx)}$ in the right-hand side, 
to obtain the equivalent condition that for all $x\in\Re^{n_x}$,
\begin{equation}\label{eq:LyapSupport}
\exists P \in \pd{n_x} : \max_{\mu \in \amb}\sum_{i=1}^k \mu_i V(f_i(x,Kx)) \leq V(x) - \ell(x,Kx).
\end{equation}

In this section, we shall assume that $\amb=\ambTV(\hat{p})$. Since $\amb$ is then a polytope, it has a finite set of extreme points, 
that is, $\amb = \conv\{a^l\}_{l=1}^{n_\amb}$. Since the maximum of a convex
function over a polytope is attained at an extreme point~\cite[Thm. 32.2]{rockafellar2015convex}, 
\eqref{eq:LyapSupport} is equivalent to
\(
      \sum_{i=1}^{\dimProbSpace} a^l_i V(f_i(x,Kx))
{}\leq{}
      V(x) - \ell(x,Kx)
\)
for all $l\in\N_{[1, n_{\amb}]}$.
However, the enumeration of the vertices of 
$\amb$ is typically computationally intensive and $n_{\amb}$ grows rapidly with 
$\dimProbSpace$ (see \Cref{sec:example-methods} for timings).

We therefore present a methodology for the determination of a gain $K$ and a matrix
$P$ that satisfies \eqref{eq:LyapSupport} for the $\ell_1$-ambiguity set $\ambTV_r(\hat{p})$,
without enumerating its vertices. 
This methodology is based on the following lemma.
\begin{lem} \label{lem:reform-exact}
Let $v(x)=(v_1(x),\ldots,v_{\dimProbSpace}(x))$ with 
\[
  v_i(x) {}\dfn{} V(f_i(x,Kx)){}={} \trans{x} \trans{(A_i+B_i K)} P (A_i+B_i K) x,
\]
for $x \in \Re^{n_x}$, and let $\hat{p} \in \simplex_{\dimProbSpace}$ denote the empirical estimate~\eqref{eq:empirical-estimate}.
For an $\ell_1$-ambiguity set of radius $r$ around $\hat{p}$, i.e., $\amb = \ambTV_r (\hat{p})$,
the distributionally robust stability condition~\eqref{eq:LyapSupport} 
is equivalent to the existence of $\dimProbSpace$ functions $z_i{}:{}\Re^{n_x} \rightarrow \Re$, such that
for all $i,j\in\N_{[1,k]}$ and $x\in\Re^{n_x}$,
\begin{equation}\label{eq:LyapSupportConv4}
      v_i(x) - z_i(x) \pm r z_j(x) + \trans{z(x)}\hat{p}
{}\leq{} 
      V(x)-\ell(x,Kx).
\end{equation} 
\end{lem}
\begin{proof}

The left-hand side of the inequality in~\eqref{eq:LyapSupport} is equivalent to the definition of 
the support function $\sigma_\amb(v(x))$ of $\amb$, evaluated at $v(x)$. 
Computing $\sigma_\amb (v(x))$ directly is seemingly not an easy task.
However, $\amb$ can be written as the 
intersection of two sets with easily computable support functions: 
\(
 \amb = \simplex_{\dimProbSpace} \cap C,
\)
where $C \dfn \ball_1(\hat{p},r)$. In fact,
 \begin{subequations}\label{eq:support_fcns}
 \begin{align}
 \sigma_{\simplex_{\dimProbSpace}}(v)&=\max\{v_1,\ldots,v_\dimProbSpace\}\\
 \sigma_{C}(v)&=r\|v\|_\infty+ \trans{v} \hat{p}.
 \end{align}
 \end{subequations}
By~\cite[Ex.~13.3(i)]{BauschkeCombettes2017}, we have that
\[ 
\sigma_{\simplex_{\dimProbSpace} \cap C}(v) =
\delta^{*}_{\simplex_{\dimProbSpace} \cap C}(v) = 
(\delta_{\simplex_{\dimProbSpace}} + \delta_C)^*(v).
\]
Thus, by the Attouch-Br{\'e}zis theorem~\cite[Thm.~15.3]{BauschkeCombettes2017}, 
\[
      \sigma_{\simplex_{\dimProbSpace}\cap C}(v)
{}={}
      (\sigma_{\simplex_{\dimProbSpace}} {}\oblong{} \sigma_{C})(v),
\]
where $\oblong$ denotes the infimal convolution, given by
\begin{align*}
      (\sigma_{\simplex_{\dimProbSpace}}\oblong\sigma_{C})(v)
{}={}&
      \inf_{z} \sigma_{\simplex_{\dimProbSpace}}(v{}-{}z)
      {}+{}
      \sigma_{C}(z).
\end{align*}
Therefore we can equivalently express~\eqref{eq:LyapSupport} as
\begin{equation}\label{eq:LyapSupportConv}
    \inf_{z} 
      \sigma_{\simplex_{\dimProbSpace}}(v(x)-z)
{}+{}
      \sigma_{C}(z)\leq V(x)-\ell(x, Kx),
\end{equation}
for all $x\in\Re^{n_x}$.
Eq.~\eqref{eq:LyapSupportConv} is true if and only if there exists a 
\(
      z(x)
{}={}
      (z_1(x), \ldots, z_{\dimProbSpace}(x))
\) such that 
\begin{equation}\label{eq:LyapSupportConv2}
\sigma_{\Delta_{\dimProbSpace}}(v(x)-z(x)) {}+{} \sigma_{C}(z(x)) {}\leq{} V(x) {}-{} \ell(x, Kx),
\end{equation}
for all $x\in\Re^{n_x}$. 
Using~\eqref{eq:support_fcns}, we express~\eqref{eq:LyapSupportConv2} as
\begin{multline}\label{eq:LyapSupportConv3}
\max_{i\in\N_{[1,\dimProbSpace]}}
\{v_i(x)-z_i(x)\} + r\|z(x)\|_\infty + \trans{z(x)}\hat{p}\\
\leq V(x) - \ell(x, Kx)
\end{multline}
In turn, this is true if and only if
\begin{equation*}
      v_i(x) {}-{} z_i(x) {}\pm{} r z_j(x) {}+{} \trans{z(x)}\hat{p} 
{}\leq{} 
      V(x)-\ell(x,Kx)
\end{equation*}
for all $i,j\in\N_{[1,k]}$ and $x\in\Re^{n_x}$, which is exactly condition~\eqref{eq:LyapSupportConv4}.
\end{proof}
We shall proceed by assuming that the components of $z(x)$ are quadratic functions 
of $x$ of the form  $z_i(x)=\trans{x} H_ix$, where 
$H_i\in\Re^{n_x{}\times{}n_x}$ are symmetric matrices,
which allows to cast \eqref{eq:LyapSupportConv4} as a set of 
$2\dimProbSpace^2$ matrix inequalities
\begin{multline}\label{eq:matrix-inequality}
    \trans{(A_i + B_i K)} P (A_i + B_i K) - H_i \pm r H_j + \\
    \textstyle\sum_{l=1}^k \hat{p}_l H_l - P + Q + \trans{K} R K 
{}\preccurlyeq{} 
    0. 
\end{multline}
for $i,j\in\N_{[1, \dimProbSpace]}$, which can be described by an LMI 
as shown in the following proposition.
\begin{proposition}\label{prop:LMIreform}
 The matrix inequality \eqref{eq:matrix-inequality} is equivalent to the LMIs
 \[ 
        \smallmat{
             -W -\hat{H}_i \pm r \hat{H}_j + \sum_{l=1}^{k} \hat{p}_l \hat{H}_l  & W \trans{A}_i + \trans{Z} \trans{B}_i  & WQ^{\frac{1}{2}} & \trans{Z} R^{\frac{1}{2}}\\ 
             *   &  -I_{n_x}     &  0               & 0                \\
             *   &  *            & -I_{n_u}         & 0                \\ 
             *   &  *            & *                &  -W
         }
 {}\preccurlyeq{}
       0
 \]
 for all $i,j\in\N_{[1, \dimProbSpace]}$, 
 where $P {}\dfn{} W^{-1}$, for $W\in\pd{n_x}$, 
       $\hat{H}_i = WH_iW$ and $Z = KW$.
\end{proposition}

\begin{proof}
We pre- and post- multiply \eqref{eq:matrix-inequality} by $W$ to obtain, 
\begin{multline}
      \trans{(A_i W + B_i Z)} W^{-1} (A_i W + B_i Z) - W + 
\\
	-\hat{H}_i \pm r\hat{H}_j 
	+ \textstyle\sum_{l=1}^k p_l \hat{H}_l + \trans{(W Q^{\frac{1}{2}})} (W Q^{\frac{1}{2}}) -
\\
	\trans{(Z R^{\frac{1}{2}})} (Z R^{\frac{1}{2}}) 
{}\preccurlyeq{} 
      0.
\end{multline}
Now define 
\[ 
\Theta_i \dfn \smallmat{
        A_i W + B_i Z \\ 
        WQ^{\frac{1}{2}} \\ 
        ZR^{\frac{1}{2}}
    }, \; 
D \dfn \smallmat{
               W &  &\\
                      & I_{n_x+n_u}
    }
\]
to obtain 
\(
      \trans{\Theta}_i D^{-1} \Theta_i 
      {}+{} (-W-\hat{H}_i \pm r \hat{H}_j 
      {}+{} \textstyle\sum_{l=1}^{k} p_l \hat{H}_l ) 
{}\preccurlyeq{}
      0, 
\)
which, by the Schur complement lemma \cite[Sec. 2.1]{lmibook}
is equivalent to the LMI
\begin{align*}
\smallmat{
    -W - \hat{H}_i \pm r \hat{H}_j + \sum_{l=1}^{k} p_l \hat{H}_l  & \trans{\Theta}_i\\
    \Theta_i & D  
} 
{}\preccurlyeq{}&
      0,
&
      W 
{}\succ{}& 
      0,
\end{align*}
which expands to the given LMI.
\end{proof}

The assumption that the components of $z(x)$ are quadratic can be justified by noting that 
a mapping $z$ that minimizes the left-hand side in~\eqref{eq:LyapSupportConv} can be taken to be
a piecewise affine function of $v$~\cite{patrinos2011convex}. In fact, due to homogeneity of the support functions, it can be easily seen that $z$ can be taken to be piecewise linear. Therefore, $z(x)$ is piecewise quadratic and homogeneous of degree two.
However, the task of computing the exact expression of $z$ is equivalent to solving a parametric linear program, hence as complex as enumerating the vertices of $\amb$. Therefore, a sensible approximation is to impose that $z(x)$ is simply quadratic. 
Moreover, in~\Cref{sec:numerical}, we demonstrate that in practice, the induced conservativeness is limited, whereas the 
computational advantage of the reformulation in \Cref{prop:LMIreform} compared to vertex enumeration allows us to solve
problems of a significantly larger scale.  

Lastly, note that the derivation leading to~\Cref{lem:reform-exact}
is not limited to $\ell_1$-based --- or even polytopic --- ambiguity sets,
as it can be easily repeated for other ambiguity sets 
which can be described as intersections of convex sets with easily computable support functions. 

%% file: Numerical.tex

\section{Numerical experiments} \label{sec:numerical}

\subsection{Data-driven ambiguity bounds} \label{sec:numerical:tightness}
\begin{figure}[ht!]
\centering
\includegraphics{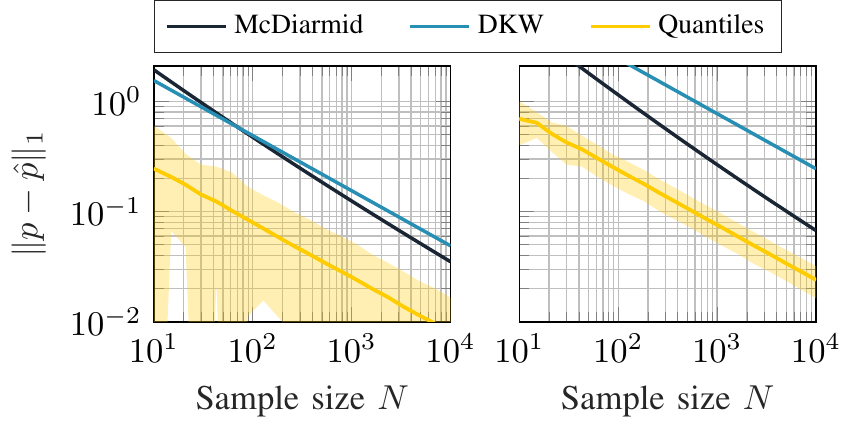}
\caption{Comparison of the derived bounds $\rmcd$ and $\rdkw$ at confidence level 
         $1-\alpha = 0.9$ for \textit{(Left)} $k = 2$ and \textit{(Right)} $k=10$. 
         Additionally, the shaded area is bounded by the empirical 
         0.1 and 0.9 quantiles of $\|p - \hat{p}\|_1$, 
         based on $10,000$ randomly generated data sets. }
\label{fig:n_vs_bounds}
\vspace{-5pt}
\end{figure}

We compare the behavior of the DKW-based radius (\cref{sec:DKW}) and the radius 
based on McDiarmid's inequality (\cref{sec:McDiarmid}) with respect to increasing sample sizes. 
\Cref{fig:n_vs_bounds} shows a comparison for two values of $\dimProbSpace$. 
Since \(\rmcd\) scales better with \(k\) (\(\mathcal{O}({k}^{\frac{3}{4}})\)) compared to $\rdkw$
(\(\mathcal{O}(k)\)), $\rdkw$ is generally lower than
$\rmcd$, especially for large values of \(k\).
However, for very low values of $\dimProbSpace$ and $N$, \Cref{fig:n_vs_bounds} demonstrates
that $\rdkw$ is tighter, albeit only by a small margin.
In practice, we may of course exploit the closed-form expressions to obtain a tighter 
bound which is simply $r = \min\{ \rmcd, \rdkw\}$.
\Cref{fig:amb_vs_N} illustrates the corresponding $\ell_1$-ambiguity sets for $\dimProbSpace=3$.

\begin{figure}[ht!]
\centering
\includegraphics[width=\columnwidth]{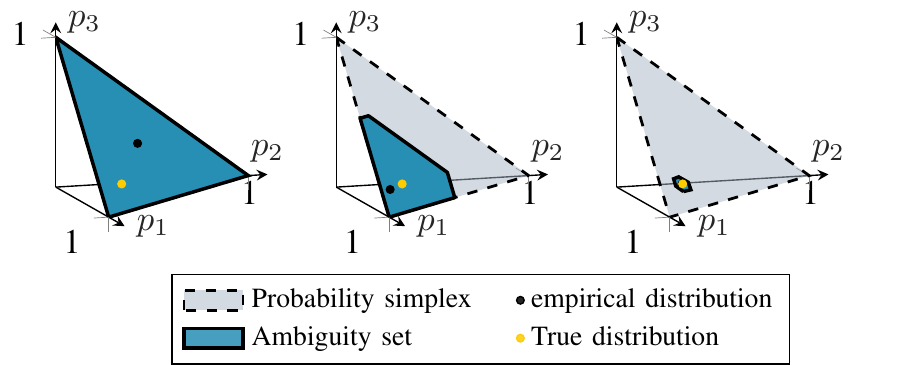}
\caption{Probability estimate $\hat{p}$ and ambiguity set $\ambTV_{\rmcd} (\hat{p})$ 
         at confidence level $1-\alpha = 0.9$ for (Left) $N=10$, 
         (Middle) $N=100$ and (Right) $N=5000$.}
\label{fig:amb_vs_N}
\vspace{-15pt}
\end{figure}

\subsection{Methods for controller design} \label{sec:example-methods}
\subsubsection{Timings}
In~\Cref{sec:controller-design}, we derived an approximation
of the Lyapunov-type stability condition~\eqref{eq:LyapSupport} which removes the 
need to solve as 
many LMIs as the number of vertices $n_\amb$ of the polytopic ambiguity set $\ambTV_r(\hat{p})$.
In \Cref{fig:timings}, we present a comparison of this approach with the 
vertex enumeration approach in terms of 
computational complexity for a system with $n_x = n_u = 2$.
For $\dimProbSpace > 7$, the vertex enumeration approach fails due to excessive
memory requirements caused by the rapid increase of $n_{\amb}$.
On the same machine, using the proposed reformulation, problems of at least 
$\dimProbSpace = 30$ could still be solved without running out of memory. 
Moreover, we observe that simply computing the vertices of $\ambTV_r(\hat{p})$ 
already proves to be more time-consuming a problem than solving the complete LMI of the 
reformulation~\eqref{eq:matrix-inequality}.    

\begin{figure}[ht!]
\centering
\includegraphics{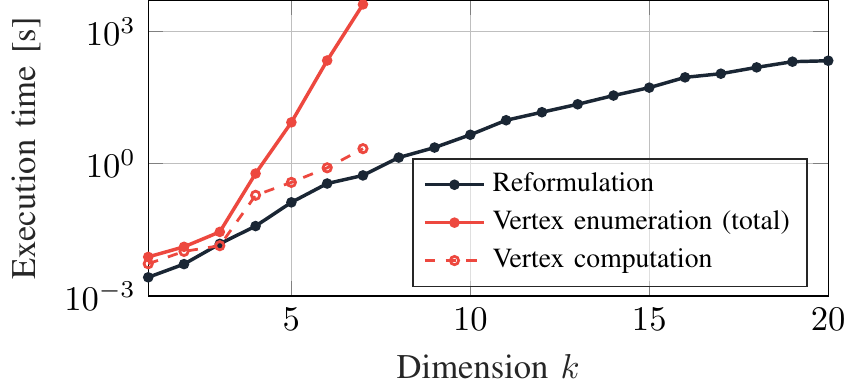}
\caption{Time to solve the LMI formulations of~\cref{eq:LyapSupport}. We compare vertex enumeration and the reformulation in~\Cref{prop:LMIreform}.
The solid lines represent the total time. For the vertex enumeration approach, the dashed line separately shows the time to compute the vertices of the ambiguity set.
(Vertex computations are performed using the MPT~\cite{mpt} toolbox, LMIs were solved using MOSEK~\cite{mosek}, on an Intel Core i7-7700K CPU at 4.20GHz.)}
\label{fig:timings}
\end{figure}

\subsubsection{Approximation quality}
We observe that in practice, the conservativeness introduced by the 
reformulation is often negligible. During experimentation, 
we have not been able to find a system for which no feasible feedback 
gain could be found through the reformulation while there could 
through vertex enumeration. This is further illustrated by the following example. Consider the system with dynamics
\[ 
    A_1 = \smallmat{0.9 & 1 \\ 0 && 0.99}, \; A_2 = \smallmat{ 1.5 & 1 \\ 0 & 2.5}, \; B_1 = B_2 = \smallmat{0\\1}.
\]
For $\hat{p}_1 = \hat{p}_2 = 0.5$, $r=0.1$, $Q=10^{-4} I_2$ and $R=10^{-4}$, we estimate the sets $\mathcal{F}$ and $\hat{\mathcal{F}}$ of feasible control gains for the exact approach (using vertex enumeration) and the reformulated LMI of \Cref{prop:LMIreform}, respectively. That is, $\mathcal{F} \dfn \{ K \in \Re^2\mid\text{\eqref{eq:LyapSupport} holds}\}$, and $\mathcal{\hat{F}} \dfn \{K \in \Re^2 \mid \exists H_i, i\in\N_{[1,\dimProbSpace]}: \text{\eqref{eq:matrix-inequality} holds} \}$. We construct a regular grid of potential feedback gains $K = [K_1 ~ K_2]$ and verify whether a $P$ (or equivalently, $W$) exists such that the involved LMI is satisfied. This point is then marked with the corresponding color in~\Cref{fig:feasibleSet}.
Since feasibility of \eqref{eq:matrix-inequality} implies feasibility of \eqref{eq:LyapSupport}, it follows that $\hat{\mathcal{F}} \subseteq \mathcal{F}$. 
We find that the experimental estimates of $\hat{\mathcal{F}}$ and $\mathcal{F}$ nearly fully overlap. In fact, in this set of 10,000 samples of $K$, only 4 instances out of 2825 that are in $\mathcal{F}$, are not in $\hat{\mathcal{F}}$.  

\begin{figure}[ht!]
\centering
\includegraphics{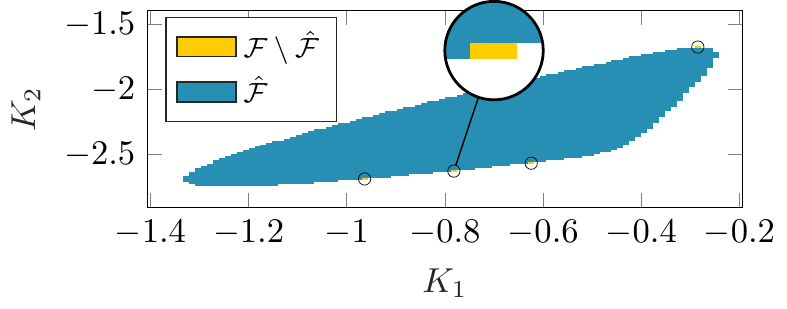}
\caption{Estimates of the feasible set $\mathcal{F}$ and $\hat{\mathcal{F}}$, defined in \Cref{sec:example-methods}.
All sampled points $K \in \mathcal{F}\setminus \mathcal{\hat{F}}$ are encircled. Infeasible points ($K \notin \mathcal{F}$) are left blank.}
\label{fig:feasibleSet}
\end{figure}


\subsection{Comparison with stochastic and robust approaches} \label{sec:numerical:vs-stochastic-robust}
The following example demonstrates  
\begin{inlinelist}
  \item the superior sample complexity of the distributionally robust approach 
  over the stochastic approach, based on the bounds obtained in~\Cref{sec:bounds}; and  
  \item the improved applicability 
  in comparison with the robust approach.   
\end{inlinelist}
This comparison is based on the \emph{distributional stability region} of 
the closed-loop system \eqref{eq:closed-loop}.
We define this as the set $\stabRegion$ of all probability vectors $p$ for which the system is MSS. Using the operator $T$, defined in~\eqref{eq:operator-stability}, 
we can denote this set as
\begin{equation} \label{eq:stabMargin-definition}
  \stabRegion \dfn \{ p \in \simplex_k \mid \spectral{(T(p))} < 1 \}.    
\end{equation}
While it is easy to test whether the system is $p$-MSS for some given $p$, 
it does not seem to be easy to determine $\stabRegion$. Indeed, since the
spectral radius of a matrix is generally not convex, aside from very specific cases, 
this set is difficult to analyze.  

However, for the following simple system
\begin{equation} \label{eq:bernoulli-sys}
A_1 = A,\, B_1=B, \; A_2 = 0, \, B_2 = 0,  
\end{equation}
which is of particular interest in networked control systems,
it is shown in~\cite{gatsis2018sample} that $\stabRegion$ for the closed-loop system with 
$u(x) = Kx$ can be written explicitly as
\begin{equation}\label{eq:stabRegionBernoulli}
\stabRegion = \left\{ p \in \simplex_2 {}\left|{}\, p_1 < \tfrac{1}{\spectral({A+BK})^2} \right. \right\}.
\end{equation}
This set simply defines a half-open line segment in $\Re^2$ and is thus convex.  
Using the convexity of this set, we may devise a simple procedure to estimate a lower bound on the confidence 
that a given linear controller is MSS for the true distribution, given only that it is stabilizing 
for $\hat{p}$, which is estimated based on $N$ \iid{} data points.
In fact, we compute \( r^{\sstar} =  \max \{ r \in [0,2] \mid \ambTV_r(\hat{p}) \subseteq \stabRegion \}\). 
Since the inclusion $\ambTV_r(\hat{p}) \subseteq \stabRegion$ can be verified easily using~\eqref{eq:stabRegionBernoulli}, $r^{\sstar}$ 
is readily computed numerically by means of a simple bisection scheme. The bounds derived in~\Cref{sec:bounds} now associate each $r^{\sstar}$
with a lower bound $(1-\alpha^{\sstar}(N))$ on the probability that a closed-loop system is $p$-MSS. We have that 
$\alpha^{\sstar}(N) = \min\{ \amcd^{\sstar}(N) , \adkw^{\sstar}(N)\}$, which, by rearranging the terms in \eqref{eq:radius-DKW} and \eqref{eq:radius-McDiarmid}, and setting $\dimProbSpace=2$, can be shown to be
\begin{align*}
\amcd^{\sstar}(N) &= e^{-\frac{N}{2} \left(\sqrt{\nicefrac{2}{\pi N}} + \nicefrac{2 \sqrt{2}}{N^{\nicefrac{3}{4}}} - r^{\ssstar}\right)^2},\\\vspace{-5pt}
\adkw^{\sstar}(N) &= 2 e^{-\frac{N (r^{\ssstar})^2}{8}}.
\end{align*}

Consider now the open-loop stochastic jump linear system of the form \eqref{eq:bernoulli-sys}, with
\[ 
    A = \smallmat{1.05 & 1.8 \\ 0 & 1.1}, B = \smallmat{1\\0},
\]
and with unknown distribution $p \in \simplex_2$. Given $N$ \iid{} observations of the disturbance $w$, we obtain an empirical probability estimate $\hat{p}$,
and a feedback gain $K$ according to the stochastic approach, i.e., the closed-loop system satisfies~\eqref{eq:LyapSupport}
for $\amb = \{ \hat{p} \}$. We compute $r^{\sstar}$, such that $\ambTV_{r^{\ssstar}}(\hat{p})$ is a tight under-approximation of $\stabRegion$.
We repeat this process for increasing values of $N$ and plot the corresponding confidence $1-\alpha^{\sstar}(N)$ that the system is $p$-MSS in~\Cref{fig:confidence-vs-samplesize}.


Similarly, to evaluate the distributionally robust approach, we compute the largest $r$, such that~\eqref{eq:LyapSupport} is feasible for $\ambTV_r (\hat{p})$ and obtain a feedback gain $K$
from solving the corresponding LMI problem. We again repeat this for increasing values of $N$ and plot the lower bound on $\prob(p \in \ambTV_r (\hat{p}))$ in \Cref{fig:confidence-vs-samplesize}.

Note that the given system is not \LRS, i.e., no linear controller exists that can stabilize the system in mean-square sense for all $p \in \simplex_2$. Therefore, the robust approach is not applicable. However, using the distributionally robust 
approach in a data-driven manner, it suffices to acquire 62 data points in order to find a controller for this same system, which is mean-square stabilizing with over 99.8\% confidence.
By contrast, obtaining similar guarantees from the stochastic approach, requires nearly 50,000 data points. 

\begin{figure}[ht!]
\centering
\includegraphics{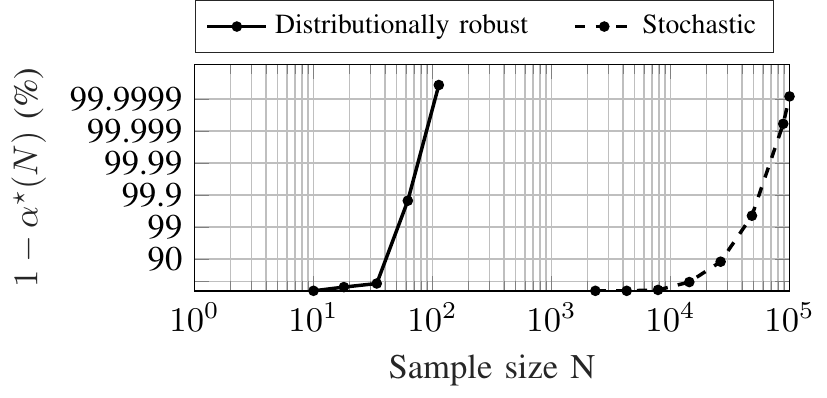}
\caption{Lower bound on the probability of obtaining a MSS controller with respect to the true distribution. For example, in order to be at least $99.8\%$ confident that the closed-loop system is mean-square stable, the stochastic approach requires a sample of about $50,000$ data points, whereas the distributionally robust approach requires merely $62$ data points.}
\label{fig:confidence-vs-samplesize}
\vspace{-10pt}
\end{figure}

%% file: conclusion.tex
\section{conclusion}
We studied the problem of data-driven synthesis of a static linear state feedback gain for stochastic jump linear systems that grants MSS \textit{with high probability}.
To this end, we adopted a distributionally robust approach, focusing specifically on $\ell_1$-ambiguity sets. 
We derived bounds that guarantee the inclusion of the true distribution in this set at the given confidence level and impose MS stability for all distributions within this ambiguity set.
To efficiently solve this problem, we derived an LMI formulation which approximates the corresponding Lyapunov-type stability condition,
but grows polynomially with the support of the dimension of the probability space. Our findings were illustrated and verified through several numerical experiments.  

In future work, we aim to generalize this methodology to Markovian disturbances and nonlinear systems.
We also aim to study the use of these results to design
terminal conditions for risk-averse risk-constrained model predictive control~\cite{sopasakis2019riskC}.

%% file: appendix.tex
\begin{appendixproof}{thm:radius-DKW} 
  \label[apdx]{proof}
  Let \(F\) denote the cumulative mass function (cmf) of $w$ and
  define \(\hat{F}\) to be the empirical cumulative distribution given $N$
  samples \(\{w_j\}_{j=1}^{N}\), that is
  \(
    \hat{F}_i \dfn \tfrac{1}{N} \sum_{j=1}^{N}1_{w_j\leq i}.
  \)
  The DKW inequality~\cite{massart1990tight} states that 
  \begin{equation} \label{eq:DKW}
      \prob\left[ \max_{i\in\mathcal{W}} \lvert \hat{F}_i - F_i \rvert > \epsilon \right] \leq 2 e^{-2N\epsilon^2}, \; \epsilon > 0,\, N \in \N.
  \end{equation}
  the cdf $F$ and the probability mass function (pmf) $p$ of a discrete distribution are related as
  \begin{align*}
  \begin{cases}
    p_1 &= F_1\\
    p_i &= F_i - F_{i-1},\; \text{for }i \in \N_{[2,\dimProbSpace]}. 
  \end{cases}
  \end{align*}
The same relation holds between the empirical counterparts \(\hat{p}\) and \(\hat{F}\). 
Therefore, a bound of the form $\lvert \hat{F}_i - F_i \rvert \leq \epsilon$ implies that, for $i > 1$ 
\begin{align*}
    \lvert \hat{p}_i - p_i\rvert &= \lvert \hat{F}_i - \hat{F}_{i-1} - (F_i - F_{i-1}) \rvert \\
                                 & \leq \lvert \hat{F}_i - F_{i} \rvert + \lvert \hat{F}_{i-1} - F_{i-1} \rvert  \leq 2 \epsilon.
\end{align*}
For $i = 1$, this inequality trivially holds as well.
Thus, writing \eqref{eq:DKW} in terms of the pmf, we obtain 
\begin{align*}
    \prob\left[ {\textstyle \max_{i\in\mathcal{W}}} \lvert \hat{p}_i - p_i \rvert > 2 \epsilon \right] &\leq 2 e^{-2N\epsilon^2}\\
    \Rightarrow \prob\left[ \| \hat{p} - p \|_1 > 2 \dimProbSpace \epsilon \right] &\leq 2 e^{-2 N \epsilon^2}.
\end{align*}
Define $r \dfn 2 \dimProbSpace \epsilon$, and \(\alpha \dfn 2e^{-2 N \epsilon^2}\) to obtain~\cref{eq:radius-DKW}.
\end{appendixproof}